\documentclass[12pt]{article}
\usepackage{fullpage}
\usepackage{graphicx,amssymb,amsmath,amsthm}
\usepackage{mathtools}
\usepackage{xspace}
\usepackage{color}
\usepackage{thmtools}
\usepackage{thm-restate}

\newtheorem{theorem}{Theorem}

\newtheorem{lemma}[theorem]{Lemma}

\newtheorem*{theorem*}{Theorem}
\newtheorem*{lemma*}{Lemma}

\newcommand{\spread}{\Delta}

\newcommand{\aspect}{\mathrm{aspect}}

\newcommand{\e}{\varepsilon}

\newcommand{\R}{\mathbb{R}}

\newcommand{\dist}{\mathbf{d}}
\newcommand{\pdist}{\mathbf{\pi}}

\newcommand{\nerve}{\mathrm{Nrv}}
\newcommand{\dgm}{\mathrm{Dgm}}
\newcommand{\pow}{\mathrm{Pow}}

\newcommand{\vor}{\mathrm{Vor}}

\newcommand{\ball}{\mathrm{ball}}

\newcommand{\radius}{\mathrm{radius}}
\newcommand{\birth}{\mathrm{birth}}

\title{A Sparse Delaunay Filtration}
\author{Donald R. Sheehy}
\begin{document}
  \maketitle

  \begin{abstract}
  We show how a filtration of Delaunay complexes can be used to approximate the persistence diagram of the distance to a point set in $\R^d$.
  Whereas the full Delaunay complex can be used to compute this persistence diagram exactly, it may have size $O(n^{\lceil d/2 \rceil})$.
  In contrast, our construction uses only $O(n)$ simplices.
  The central idea is to connect Delaunay complexes on progressively denser subsamples by considering the flips in an incremental construction as simplices in $d+1$ dimensions.
  This approach leads to a very simple and straightforward proof of correctness in geometric terms, because the final filtration is dual to a $(d+1)$-dimensional Voronoi construction similar to the standard Delaunay filtration complex.
  We also, show how this complex can be efficiently constructed.
\end{abstract}

  \section{Introduction}
\label{sec:introduction}

The persistent homology of the distance to a set of points $P$ in $\R^d$ describes the evolution of the topology of $\bigcup_{p\in P}\ball(p,\alpha)$ as $\alpha$ grows from $0$ to $\infty$.
It is a multi-scale description of the ``shape'' of the point set.
The theory of persistent homology has its origins in the work by Edelsbrunner et al.~\cite{edelsbrunner83shape,edelsbrunner95union} on $\alpha$-hulls and $\alpha$-shapes and their relation to the Delaunay triangulation.
% TODO: what is the difference?
The paper that introduced persistent homology~\cite{edelsbrunner02topological} was based on ordering the simplices of the Delaunay triangulation.
Many papers have followed that use alternatives to the Delaunay triangulation in different spaces, but in Euclidean space, the Delaunay triangulation has a certain perfection in its ability to minimally represent the topology of the distance function.
For approximations, the potential $O(n^{\lceil d/2 \rceil})$ size of the Delaunay triangulation cannot compete with the linear size of so-called sparse filtrations~\cite{sheehy13linear,botnan14approximating,cavanna15geometric,dey16simba}.
In this paper, we combine the ideas from sparse filtrations with the Delaunay triangulation, achieving both the elegance of the Delaunay triangulation and the worst-case linear-size guarantees.
Along the way, we will give a new topological perspective to the classic approach to computing Delaunay triangulations by flips.

A \emph{flip} in a $2$-dimensional triangulation is the replacement of two adjacent triangles whose four vertices are in convex position with the other two possible triangles on the same vertices.
A classic way of visualizing flips is to view the two configurations as projections of the upper and lower hull of a tetrahedron in three dimensions (see Fig.~\ref{fig:flips}).
This view also permits one to interpret other operations as flips, such as the insertion of a new vertex splitting one triangle into three.
We call the former class of flips $(2,2)$-flips and the latter $(1,3)$-flips, indicating the number of triangles before and after the flip.
More generally, there are $(k, d+1-k)$-flips for sets of $d+2$ points in $\R^d$.
These are likewise interpreted as projections of $(d+1)$-simplices.
In this paper, we will use the $(d+1)$-simplices of the flips to give a topological connection between the Delaunay triangulation of a set of points and the Delaunay triangulation of a subset.
Throughout, we will distinguish between the terms Delaunay triangulation for the embedded geometric complex on a set of points and a Delaunay complex which is the corresponding abstract simplicial complex.
As will become clear, the addition of the flip simplices will result in a simplicial complex that will not be embedded in $\R^d$.

\begin{figure}
  \includegraphics[width=\textwidth]{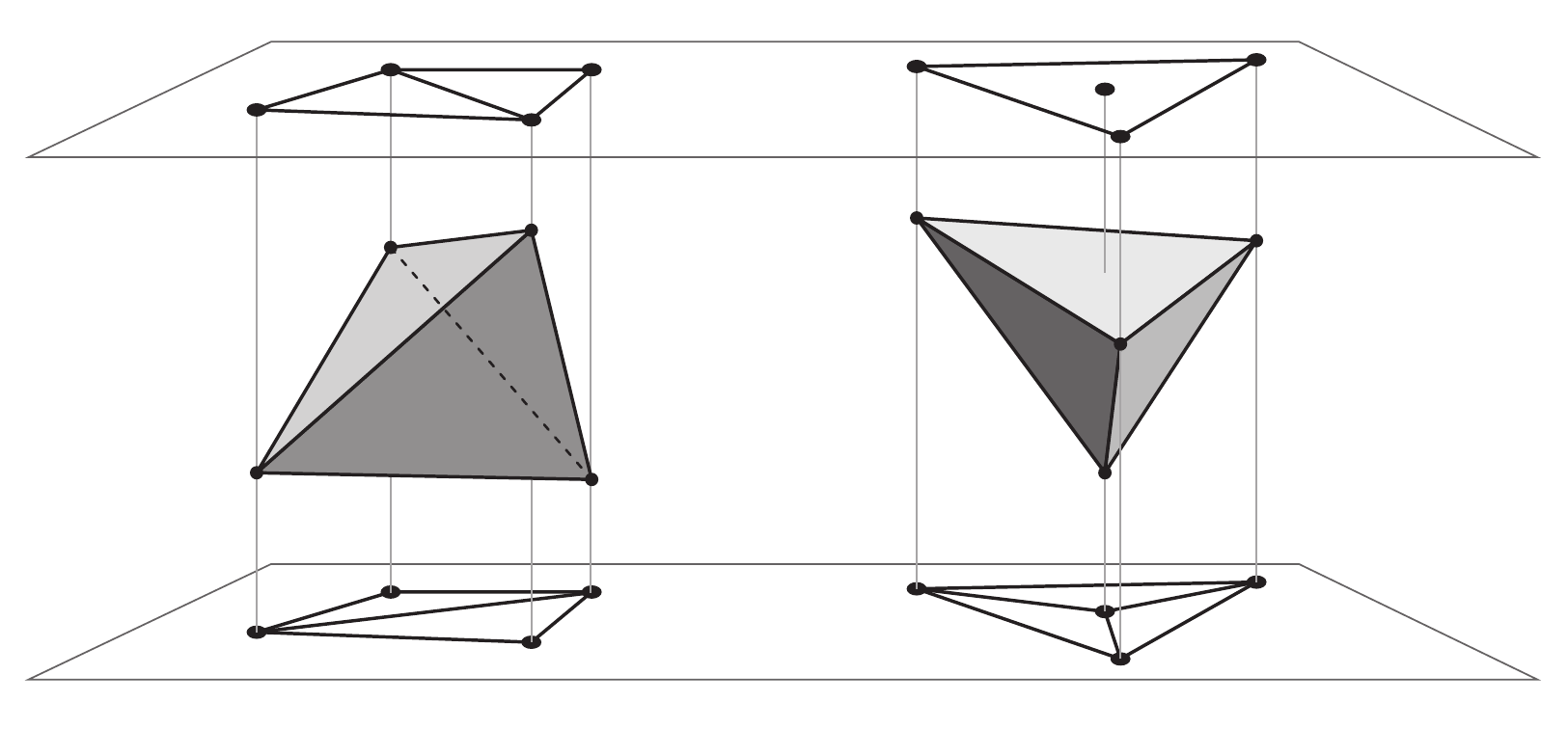}
  \caption{Flips in the plane correspond to the upper and lower facets of a tetrahedron in $\R^3$.}\label{fig:flips}
\end{figure}

Many Delaunay triangulation algorithms use flips as an algorithmic primitive.
This idea goes back to the work of Lawson~\cite{lawson72transforming,lawson77software} and reached its more modern form in the simultaneous papers of Bowyer~\cite{bowyer81computing} and Watson~\cite{watson81computing}.
As the $d$-simplices of the Delaunay triangulation are those for which the circumsphere is empty of other input points, a transformations of the problem into $\R^{d+1}$ can make this criterion linear.
This was first done by Brown~\cite{brown79voronoi} using the stereographic map and later by Edelsbrunner and Seidel~\cite{edelsbrunner86voronoi} using the parabolic lifting.
Adjusting the parabolic lifting can be interpreted as assigning weights to points and leads to \emph{weighted Delaunay triangulations} (also known as \emph{regular} triangulations).
Edelsbrunner and Shah showed that incremental, flip-based algorithms also work in the weighted case~\cite{edelsbrunner96topological} despite obstructions discovered by Joe~\cite{joe89three} to applying Lawson's algorithm in $\R^3$.

% Flipping based incremental construction of the weighted Delaunay triangulation can be
% In an incremental construction of the Delaunay triangulation, the

The theory of persistent homology applies to settings much more generally than the sublevel sets of distance functions in Euclidean space.
One common setting is to consider points in finite metric spaces.
The complex used for this is called the (Vietoris-)Rips complex.
At scale $\alpha$, it contains a simplex for every clique in the $\alpha$-neighborhood graph of the points.
In order to deal with the size blowup of this complex, several sparsification methods have been proposed~\cite{sheehy13linear,botnan14approximating,cavanna15geometric,dey16simba}.
All of these approaches attempt to use only subsets of the input points as the scale increases.
In this paper, we will show how to adapt this approach to the Delaunay complex.

  \section{Background}
\label{sec:background}

\subsection{Distances, Points, and Weights}

Let $\|a-b\|$ denote the Euclidean distance between $a$ and $b$ in $\R^d$.
Let $\ball(c,r)$ denote the closed ball centered at $c\in \R^d$ with radius $r$.
For a set $P\subset\R^d$, let
\[
  \dist(x, P)= \min_{p\in P} \|x-p\|.
\]
Equivalently, $\dist(x,P)$ is the minimum $r$ such that $P\cap \ball(x,r)$ is nonempty.

The distance function induced by a point set $P$ maps each point $x\in \R^d$ to $\dist(x, P)$.
The sublevel sets of this distance function is often used in topological data analysis as a way to extend a finite set and fill in the space between the points.
These sublevel sets are sometimes called \emph{offsets} and are formally defined for each scale $\alpha$ as
\[
  P^\alpha := \{x\in \R^d \mid \dist(x,P)\le \alpha\} = \bigcup_{p\in P}\ball(p, \alpha).
\]
The Hausdorff distance between two point sets $P$ and $Q$ is defined as
\[
  \dist_H(P,Q) := \max\{\max_{p\in P}\dist(p,Q), \max_{q\in Q}\dist(q,P)\}.
\]
Equivalently, $\dist_H(P,Q)$ is the minimum $r$ such that $P\subseteq Q^r$ and $Q\subseteq P^r$.

One way to modify a distance function and give more or less importance to certain points is to assign a weight to each point.
A weighted point $\hat{p}$ is a point $p\in \R^d$ and a weight $w_p\in \R$.
The \emph{weighted distance to $\hat{p}$} is defined as
\[
  \pdist_p(x) := \sqrt{ \|x-p\|^2 + w_p^2}.
\]

The weighted distance is also called the power distance, especially in the case where one subtracts the weight rather than adds it.
Throughout the paper, we add weights rather than the usual power distance as it substantially simplifies both the conceptual use of weights to decrease the importance (i.e., radius) of some points and also the arithmetic.
All of the constructions we present can be translated into power distances by a global transformation of all the weights and a reinterpretation of the scale.
A similar approach to weighting point may be found in~\cite{buchet15efficient}.

We use the same notation for the weighted distance of a (non-weighted) point $x$ to a weighted point set $\hat{P}$ as we did in the unweighted case.
\[
  \dist(x, \hat{P}) := \min_{\hat{p}\in \hat{P}} \pdist_{\hat{p}}(x)\}.
\]
Thus, unweighted points way be viewed as points with weight zero.
The offsets of $\hat{P}$ are
\[
  \hat{P}^\alpha = \{x\in \R^d \mid \dist(x, \hat{P})\}.
\]

\subsection{Persistent Homology}

A family of subsets $\{X^\alpha\mid \alpha \in \R\}$ of $\R^d$ is called a \emph{filtration} if for all $\alpha \le \beta$, we have $X^\alpha \subseteq X^\beta$.
We will reserve superscripts on sets as a notation for filtration parameters and will denote a filtration $(X^\alpha)$ with parentheses to stress the importance of the ordering.
For filtrations that are defined only over an interval $[s,t]\subset \R$, we will assume that $X^\alpha = X^s$ for $\alpha<s$ and $X^\alpha = X^t$ for $\alpha >t$.

The \emph{persistent homology} of $(X^\alpha)$ is a representation of the changes in the topology of $X^\alpha$ as $\alpha$ varies over $\R$.
The result is a \emph{persisence diagram}, denoted $\dgm(X^\alpha)$, that is a multiset of pairs $(b,d)$ in the extended plane $(\R\cup \infty)^2$.
Each pair $(b,d)$ represents a nontrivial homology class that exists only in $X^\alpha$ for $\alpha$ in the half open interval $[b,d)$.
Thus, $b$ is the \emph{birth time} of a topological feature and $d$ is its \emph{death time}.

Persistent homology is usually computed on combinatorial objects called simplicial complexes.
A \emph{simplicial complex} is a pair of sets $(V,S)$ where $V$ is the vertex set and $S\subseteq \pow(V)$ is the simplex set.
It is required that $S$ be closed under subsets, i.e., if $\sigma\subseteq \tau \in S$, then $\sigma\in S$.
A filtration $(K^\alpha)$ of simplicial complexes is called a \emph{filtered simplicial complex}.
There is a standard way to relate the persistence diagram of subsets of $\R^d$ to the persistence diagram of a filtered simplicial complex.

% TODO: what do we mean by approximation

The main problem addressed in this paper is the efficient approximation of $\dgm(P^\alpha)$ by constructing a linear size filtered simplicial complex based on the Delaunay triangulation.

\subsection{Greedy Permutations}

For ranges of indices, let $[a:b]$ denote $\{a,\ldots b-1\}$ and $[b]$ denote $[0,b]$.
For an ordered set $P$, let $P_i$ denote the $i$th prefix $\{p_j \mid j\in [i]\} = \{p_0, \ldots, p_{i-1}\}$.

A \emph{greedy ordering} or \emph{greedy permutation} of $P$ is defined as follows.
The first point, $p_0$, may be chosen arbitrarily.
The $i$th point, $p_i$ is chosen to be a point that maximizes $\dist(p_i, P_i)$.
That is, each point after the first is the farthest from its predecessors.
Equivalently,
\[
  \dist(p_i, P_i) = \dist_H(P_i, P).
\]

Greedy permutations have been reinvented several times, especially in the context of $k$-center clustering (see Gonzalez~\cite{gonzalez85clustering} or Dyer and Frieze~\cite{dyer85simple}).
Clarkson adapted his $\mathrm{sb}$ nearest neighbor search data structure to compute greedy permutations.
Har-peled and Mendel~\cite{har-peled06fast} showed that Clarkson's approach yields an $O(n\log \spread)$-time algorithm in low-dimensional metric spaces.
They also gave an $O(n\log n)$-time algorithm in such cases.

The \textit{insertion radius} of a point $p_i$ is defined as
\[
  r_i := \dist(p_i, P_i)
\]
The definition of a greedy ordering directly implies that if $i < j$, then $r_i \ge r_j$.

Every prefix $P_i$ of a greedy permutation satisfies both packing and covering conditions in the following sense.
The set $P_i$ is a \emph{$r_{i-1}$-packing}: for every pair of points in $a,b \in P_i$, we have $\|a-b\|\ge r_{i-1}$.
The set $P_i$ is an \emph{$r_i$-covering} of $P$: for every point $a\in P$, there is a point $b\in P_i$ such that $\|a-b\|\le r_i$.
% Both of these properties follow immediately from the definition.

\subsection{Weights and Time}

We will be considering weighted point sets in which the weights vary in time.
For each point $p_i$ in $P$, we will assign a nonnegative weight function $w_i:\R\to \R$.
For a given $\alpha$, the set
\[
  \hat{P}(\alpha) := \{(p_i, w_i(\alpha) \mid p_i \in P)
\]
is a weighted point set in which the weight of $p_i$ is $w_i(\alpha)$.

The weight functions will be defined in terms of the \emph{freezing time} $\lambda_i$ of each point $p_i$.
The exact value chosen for $\lambda_i$ will depend on our desired approximation guarantees and the specifics of the algorithm.
Once the freezing times are fixed, the weights are defined as follows.
\[
  w_i(\alpha) = \begin{cases}
    0 & \text{if }\alpha<\lambda_i\\
    \sqrt{\alpha^2 - \lambda_i^2} &\text{otherwise}\\
    \end{cases}
\]
The impact of this weight function on distances is most clearly seen by considering a ball of radius $\alpha$ at a point.
The squared power distance of a point $x$ to a point $q$ with weight $w$ is $\|x-q\|^2 + w^2$.
It follows that the ball $b_i(\alpha)$ centered at $p_i$ with weight $w_i(\alpha)$ is
\[
  b_i^\alpha = \ball(p_i, \min\{\alpha, \lambda_i\}).
\]
This is why $\lambda_i$ is called the freezing time; at scales $\alpha> \lambda_i$, the Euclidean radius of a ball of weighted radius $\alpha$ will not grow.

The following lemma shows how weighting the points according to the freezing time guarantees that at all scales, there is always a point nearby that is sufficiently close and sufficiently far from its freezing time.
The proof can be found in Appendix~\ref{sec:covering_lemmas}.

\begin{restatable}{lemma}{weightedcover}\label{lem:weighted_cover}
  Let $P\subset \R^d$ be ordered according to a greedy permutation with insertion radii $r_0\ldots r_{i-1}$.
  Let $\e>0$ be any constant.
  For all $j\in [1:n]$, let the freezing times $\lambda_j$ be chosen so that $\lambda_j \ge \frac{1+\e}{\e}r_j$.
  Then, for all $k\in [1:n]$ and all $\alpha\ge 0$, there exists $i$ such that
  \begin{itemize}
    \item $\lambda_i \ge (1+\e)\alpha$, and
    \item $\|p_i - p_k\| < \e \alpha$.
  \end{itemize}
\end{restatable}

If we take the special case of $\alpha = \lambda_k$ in the preceding lemma, we see that $\ball(p_k, \lambda_k)$ is completely covered by $\ball(p_i, (1+\e)\lambda_k)$, where the point $p_i$ is not yet frozen at time $(1+\e)\lambda_k$.
For larger values of $\alpha$, there will always be a point $p_i$ to cover this ball.
This lemma has two important consequences.
First, it implies that we will be able to remove or ignore the point $p_k$ after time $(1+\e)\lambda_k$.
% TODO: future reference to Voronoi lemma of some kind
Second, it allows us to relate the offsets of $P$ with the weighted offsets as shown in the following lemma whose proof may be found in Appendix~\ref{sec:covering_lemmas}.

\begin{restatable}[Weighted Offset Interleaving]{lemma}{weightedoffsetinterleaving}\label{lem:weighted_offset_interleaving}
  Let $P\subset \R^d$ be ordered according to a greedy permutation with insertion radii $r_0 \ldots r_{i-1}$.
  Let $\e>0$ be any constant.
  For all $j\in [1:n]$, let the freezing times $\lambda_j$ be chosen so that $\lambda_j \ge \frac{1+\e}{\e}r_j$.
  Then, for all $\alpha\ge 0$,$\hat{P}^{\alpha} \subseteq P^\alpha \subseteq \hat{P}^{(1+\e)\alpha}$.
\end{restatable}

\subsection{Delaunay and Voronoi}\label{sec:delaunay_and_voronoi}

\begin{figure}[h]
  \centering
  \includegraphics[width=0.29\textwidth]{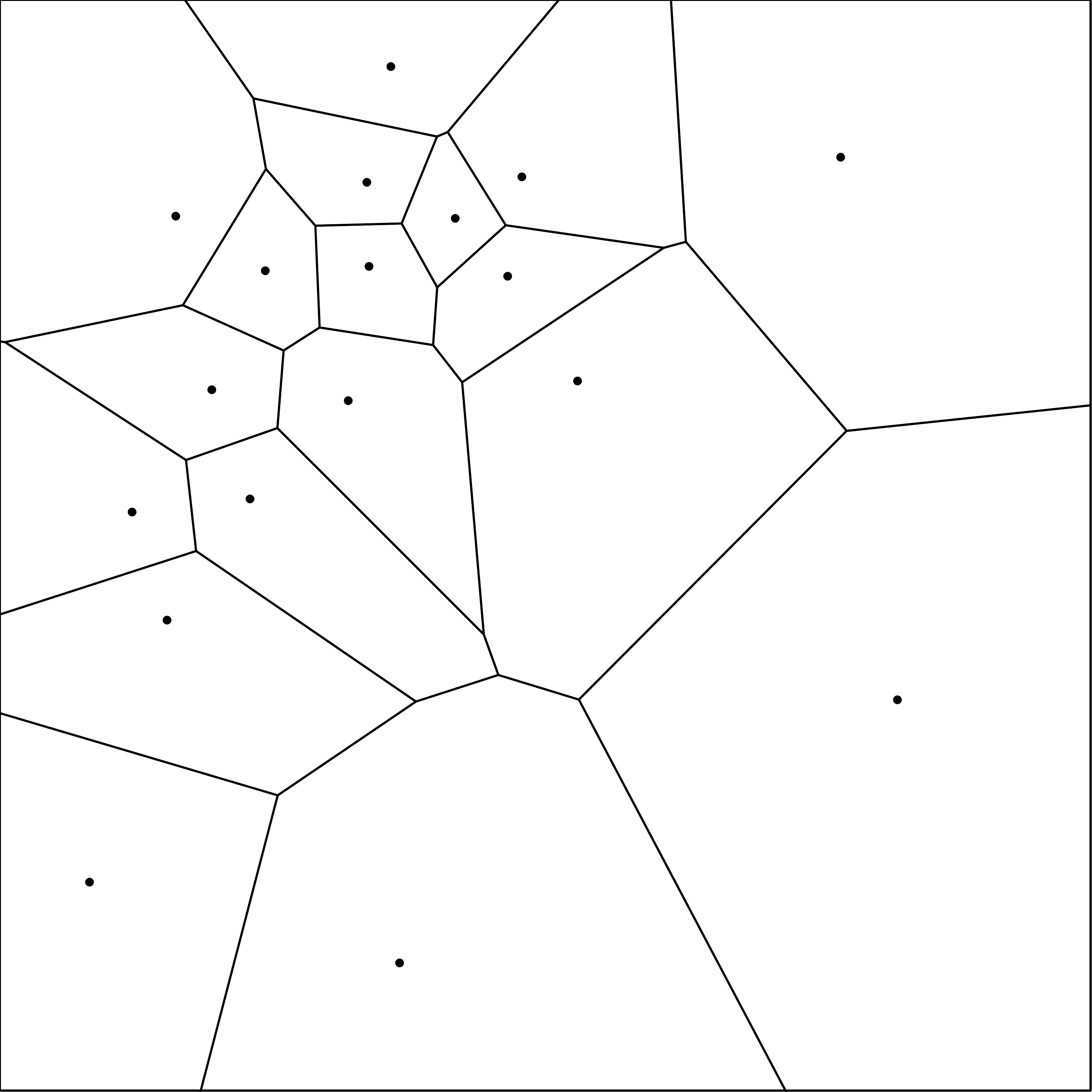}
  \includegraphics[width=0.29\textwidth]{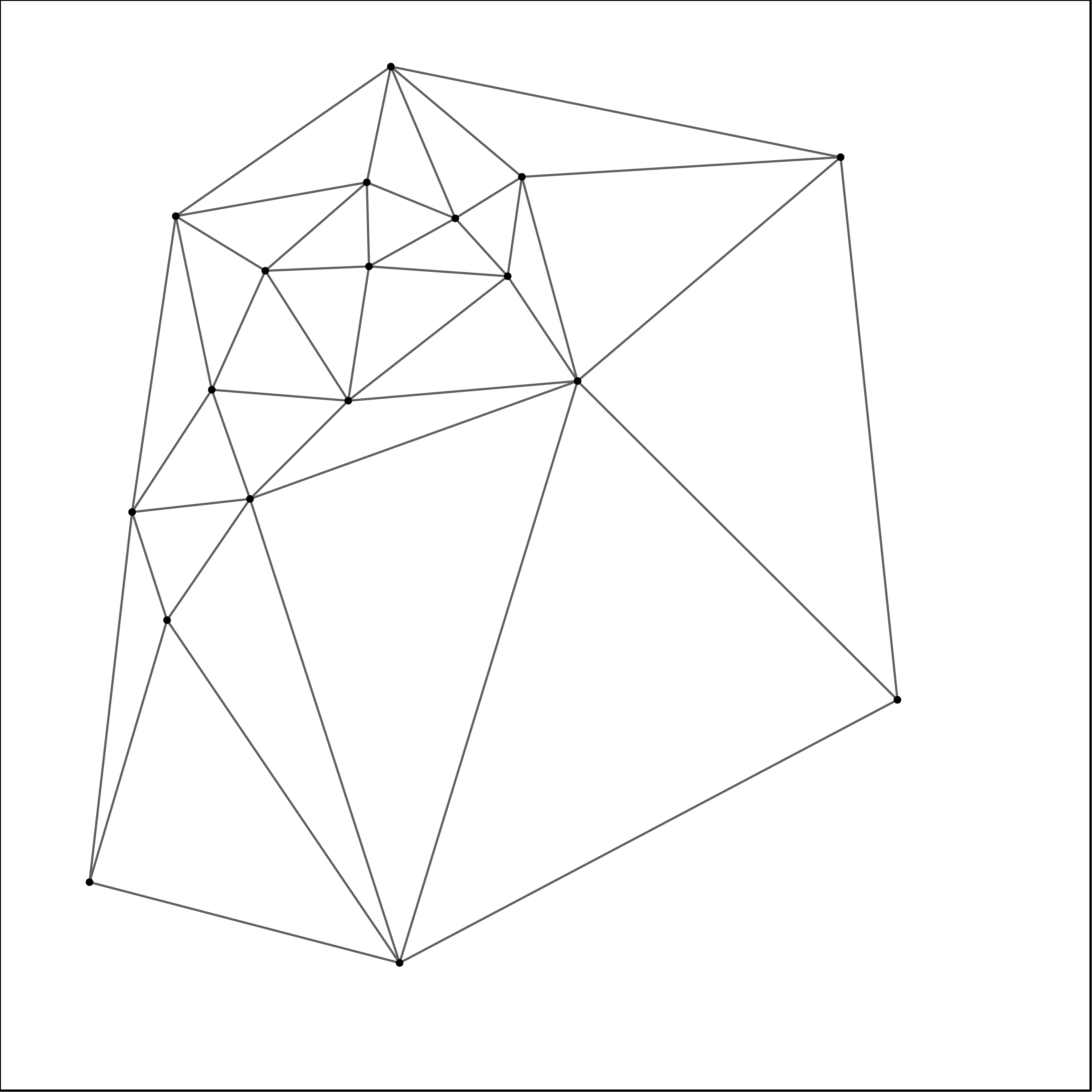}
  \caption{The Voronoi diagram and its dual Delaunay triangulation.}
  \label{fig:algorithm}
\end{figure}

Let $P$ be a set of points in $\R^d$.
The \emph{Voronoi cell} of a point $q\in P$ is defined as
\[
  \vor_P(q) := \{x\in \R^d \mid \|x-q\| = \dist(x, P)\}.
\]
For a subset of points $S\subseteq P$, we can define its Voronoi cell as
\[
  \vor_P(S) := \bigcap_{q\in S} \vor_P(q).
\]
The same definitions apply equally to the case of weighted points.
% The kind of weighted distance function considered in this paper yield what are also called \emph{power diagrams}.
It is well-known that the Voronoi cells are polyhedra.
The \emph{Voronoi diagram} of $P$ is defined as the polyhedral complex composed of nonempty Voronoi cells $\vor_P(S)$ for all $S\subseteq(P)$.

The \emph{Delaunay complex} (also known with some nuances as the Delaunay triangulation, tesselation, or mosaic) is the simplicial complex formed by the subsets $S\subseteq P$ for which $\vor_P(S)$ is nonempty.
The subsets are the \emph{simplices} and the \emph{dimension} of $S$ is $|S| -1$.
We are defining the Delaunay complex here as an abstract simplicial complex.
In the special case where all simplices have dimension at most $d$, the Delaunay complex will embed neatly into $\R^d$ with the vertices embedded at the points of $P$ and each simplex embedded as the convex closure of its vertices.
For the purposes of this paper we will not need the embedding and thus will have no need for the usual general position conditions as would usually be required for the geometric realization of the complex in $\R^d$.
In fact, we will explicitly construct ``degenerate'' Delaunay complexes because the adjustment of weights over time will necessarily pass through instants where higher dimensional simplices are present in the Delaunay complex.
These are the moments when a flip occurs.
We will only require that at most one flip occurs at a time.

% The Delaunay complex is explicitly defined by which subsets of Voronoi cells have a nonempty intersection.
% This is a special case of a standard construction in algebraic topology called a \emph{nerve}.

\subsection{The Kinetic View of Flips}

Kinetic data structures~\cite{guibas98kinetic} generalize the classic sweepline approach of Bentley and Ottmann.
The goal is to maintain some geometric structure as points move along trajectories.
The principal technique is to rewrite the geometric predicates defining the structure as functions (usually polynomials) of time and then solving (finding roots) for the time when the predicate will no longer hold.
At that time, some combinatorial change is required.
These changes are stored in a priority queue, ordered by time.

Incremental Delaunay triangulation can be phrased as a kinetic data structures problem if one understand the motion in $d+1$ dimensions as a continuous change in the weight of the point being inserted.
This perspective is often abandoned because the precise ordering of the flips is rarely important and is not necessary for correct computation (see Edelsbrunner~\cite{edelsbrunner96topological}).
However, in our case, we want the precise order and time of the flips that occur, because these inform the final filtration.
Also, we will be adding multiple points at once, so the order becomes more important.

A similar approach was used by Miller and Sheehy~\cite{miller14new} in an output-sensitive algorithm for computing Delaunay triangulations.
In that paper, it was observed that the predicate polynomials are linear for the case where the points are partitioned into two sets, one with weight zero and one with squared weight varying linearly in time.
% Miller and Sheehy also showed that it is possible to count the flips that occur purely in terms of intersections between certain cells of the Voronoi diagrams before and after the insertions.
% TODO: does that inform our analysis

\subsection{Clipping the Voronoi Diagram}

The \emph{clipped Voronoi cell} is the intersection of a Voronoi cell and a ball.
Let $\hat{P}$ be a weighted point set with weights varying in time as above.
Then, for each $p_i\in \hat{P}$ and each $\alpha \ge 0$ there is a ball
\[
  b_i(\alpha) := \ball(p_i, \min\{\lambda_i, \alpha\}).
\]
Specifically, for $p_i\in P$, we define the clipped Voronoi cell of $p_i$ as
\[
  V_i^\alpha := \vor_{\hat{P}(\alpha)}(p_i) \cap b_i(\alpha).
\]

\begin{figure}
  \centering
  \includegraphics[width=0.19\textwidth]{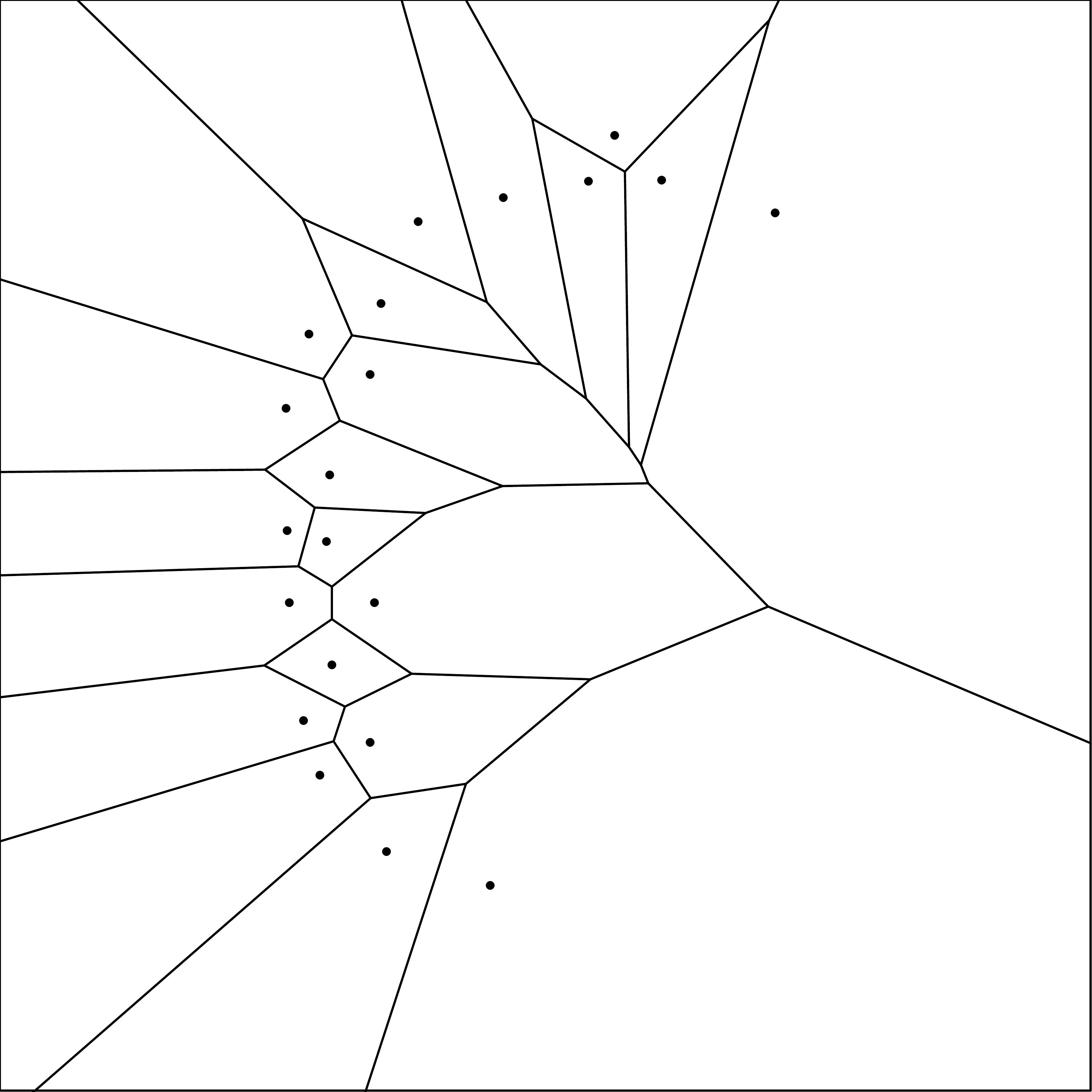}
  \includegraphics[width=0.19\textwidth]{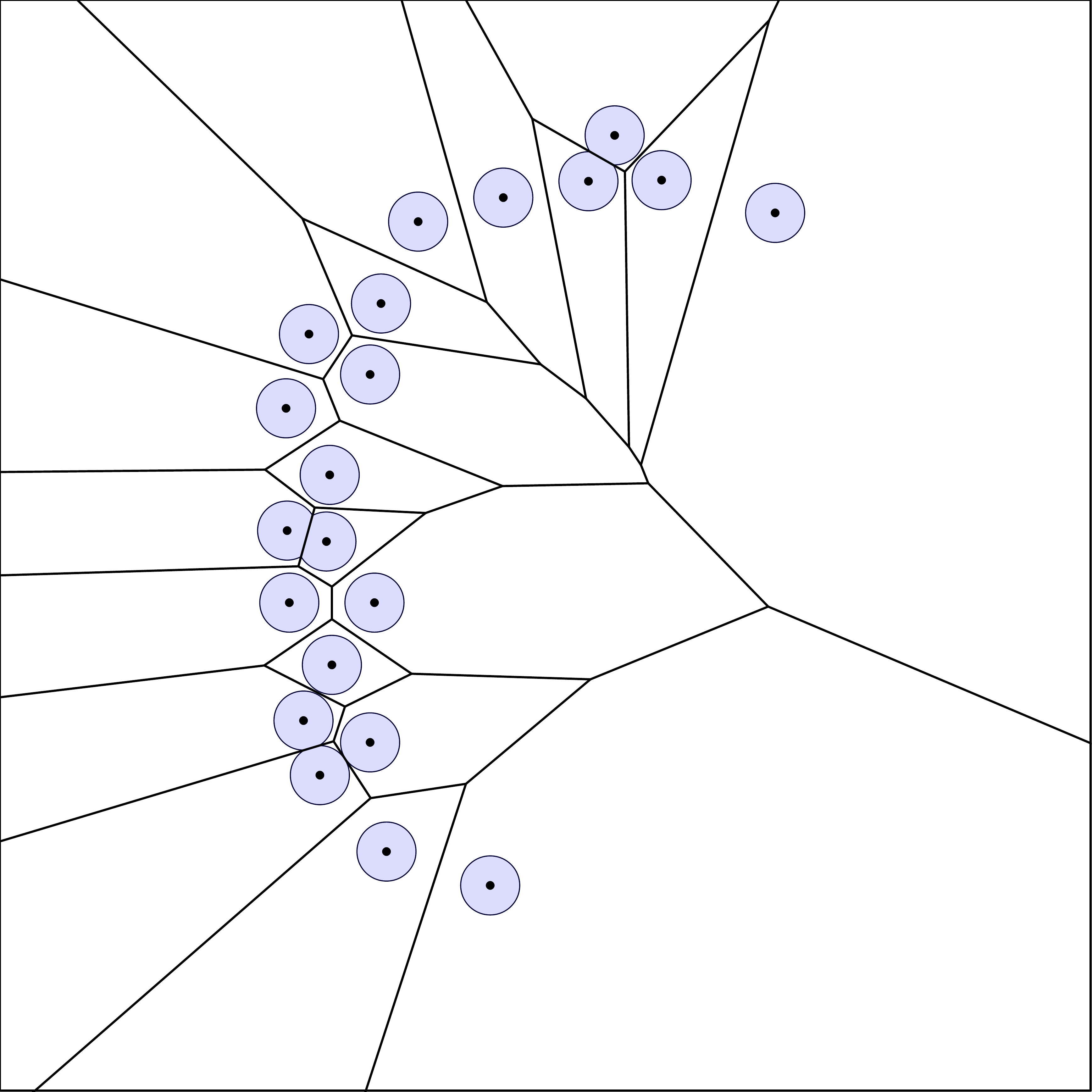}
  \includegraphics[width=0.19\textwidth]{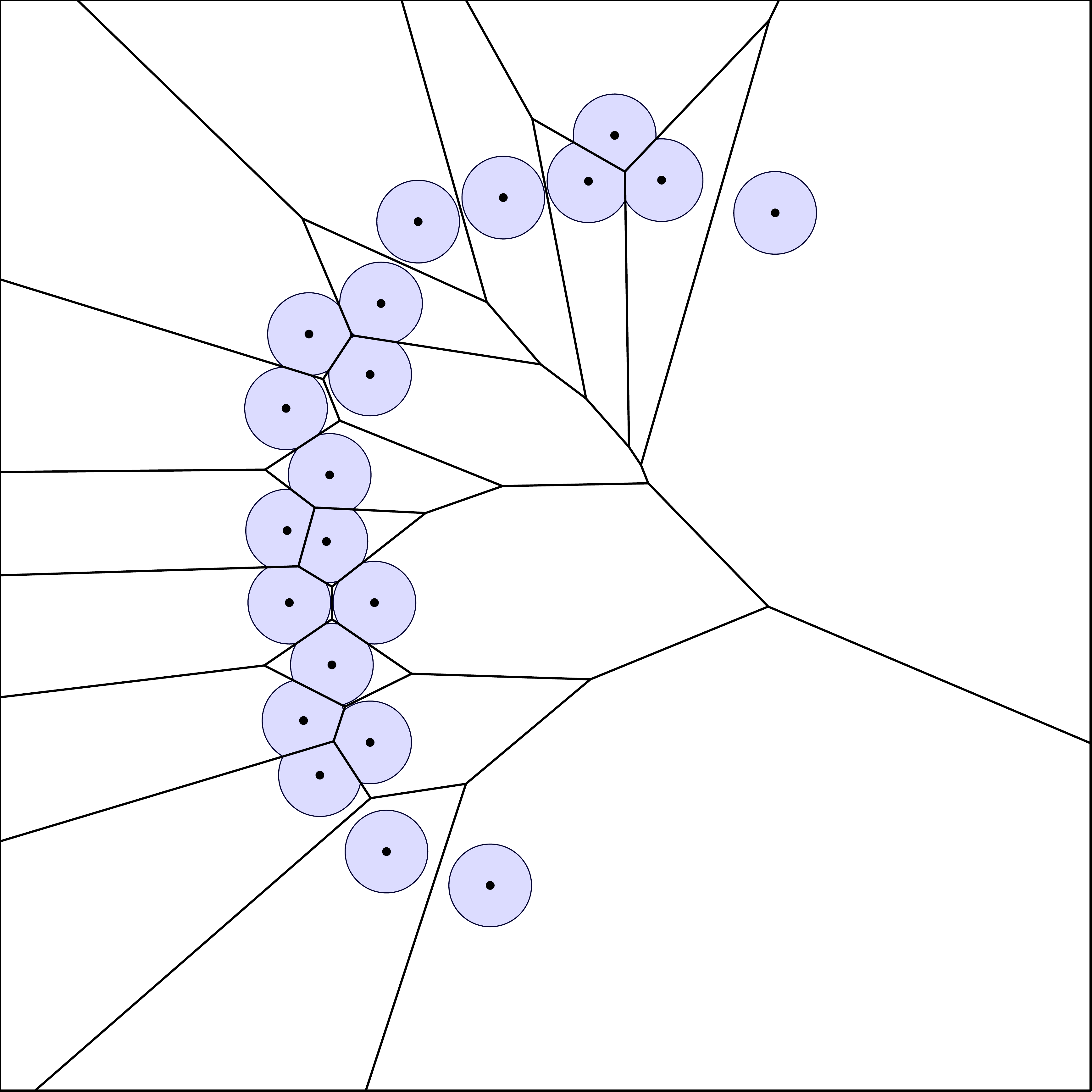}
  \includegraphics[width=0.19\textwidth]{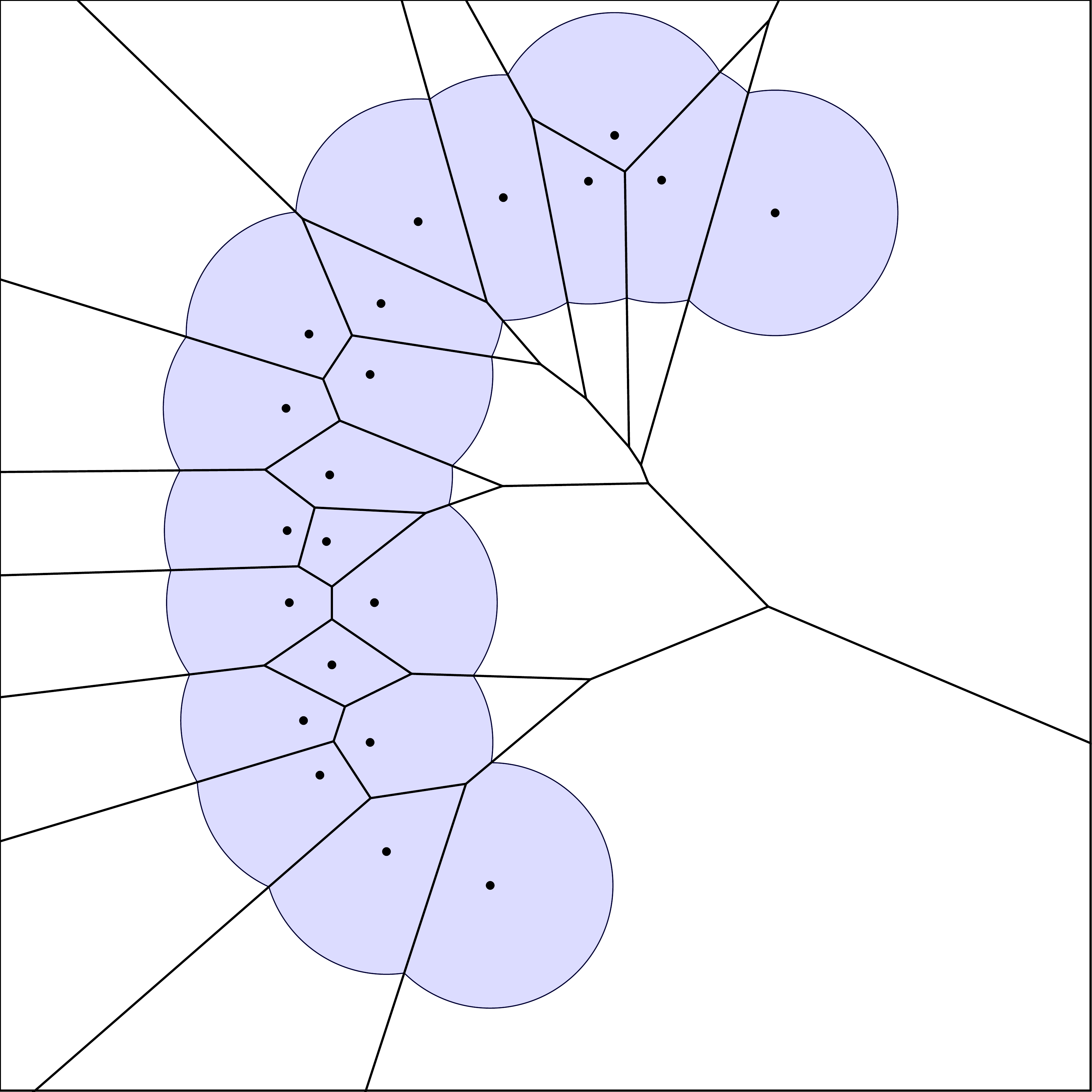}
  \includegraphics[width=0.19\textwidth]{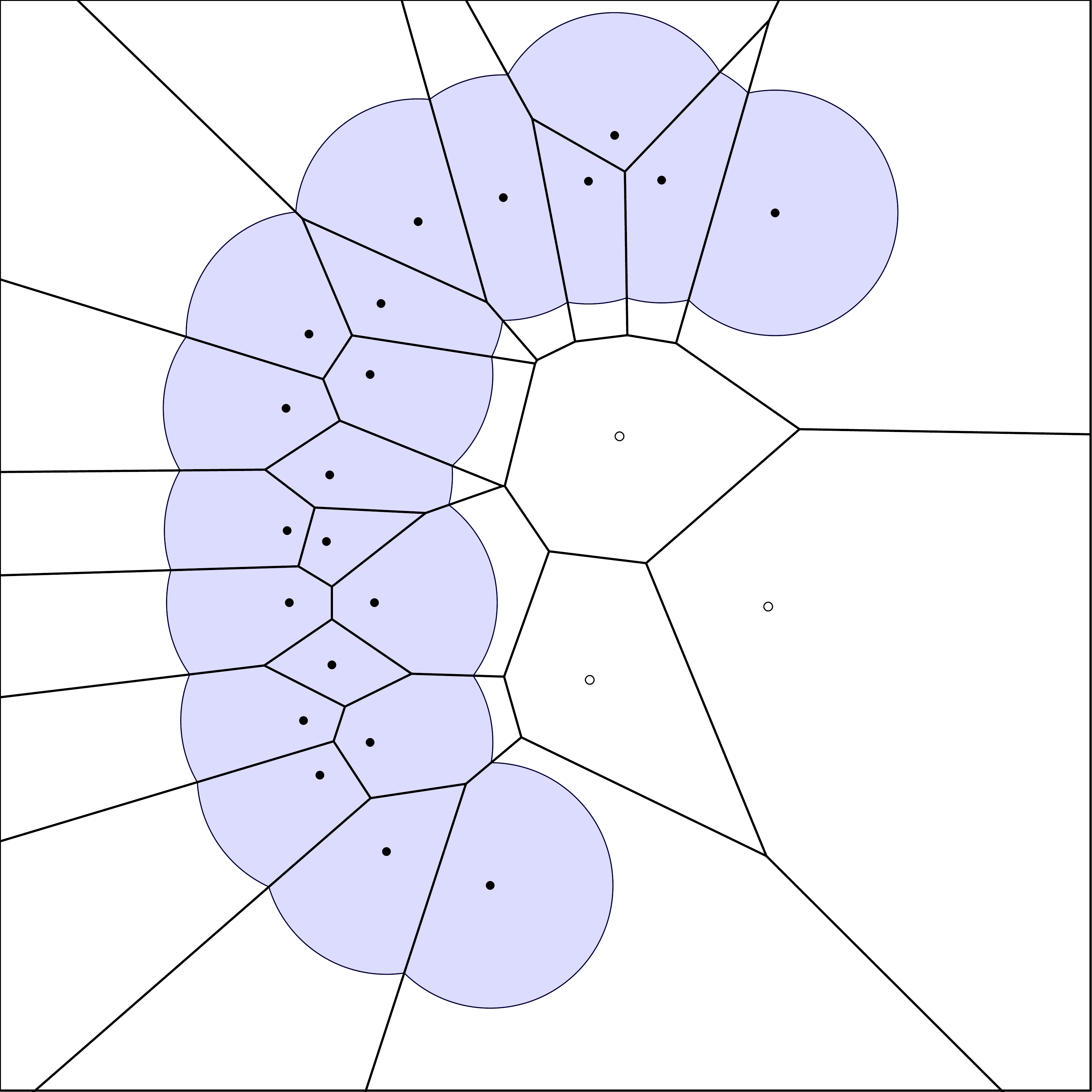}
  \caption{The clipped Voronoi cells exactly cover the offsets.  In the last frame, we illustrate how the addition of extra points may not change the offsets or the Delaunay filtration.  These extra points are used in Section~\ref{sec:algorithm} to keep the complexity linear.}
  \label{fig:algorithm}
\end{figure}

The \emph{Delaunay complex at scale $\alpha$} is the subcomplex of the Delaunay complex defined by using the clipped Voronoi cells instead of the full Voronoi cells.
That is,
\[
  D^\alpha := \{\sigma\subseteq P \mid \bigcap_{p_i\in\sigma} V_i^\alpha \neq \emptyset\}.
\]
For any $\alpha\le \beta$, we have $D^\alpha\subseteq D^\beta$, i.e., $(D^\alpha)$ is a filtration.
% In particular, in the absence of weights, $D^0$ is an isolated set of vertices and $D^\infty$ is the full Delaunay complex.

By defining the weights as above, we guarantee that for all scales $\alpha\ge (1+\e)\lambda_k$, the clipped Voronoi cell $V_k^\alpha$ will be empty.
This is a direct consequence of Lemma~\ref{lem:weighted_cover}, but we give the formal statement below and the proof in Appendix~\ref{sec:covering_lemmas} for completeness.
See Fig.~\ref{fig:lifted_clipped_cells} for an illustration.

\begin{figure}
  \includegraphics[width=\textwidth]{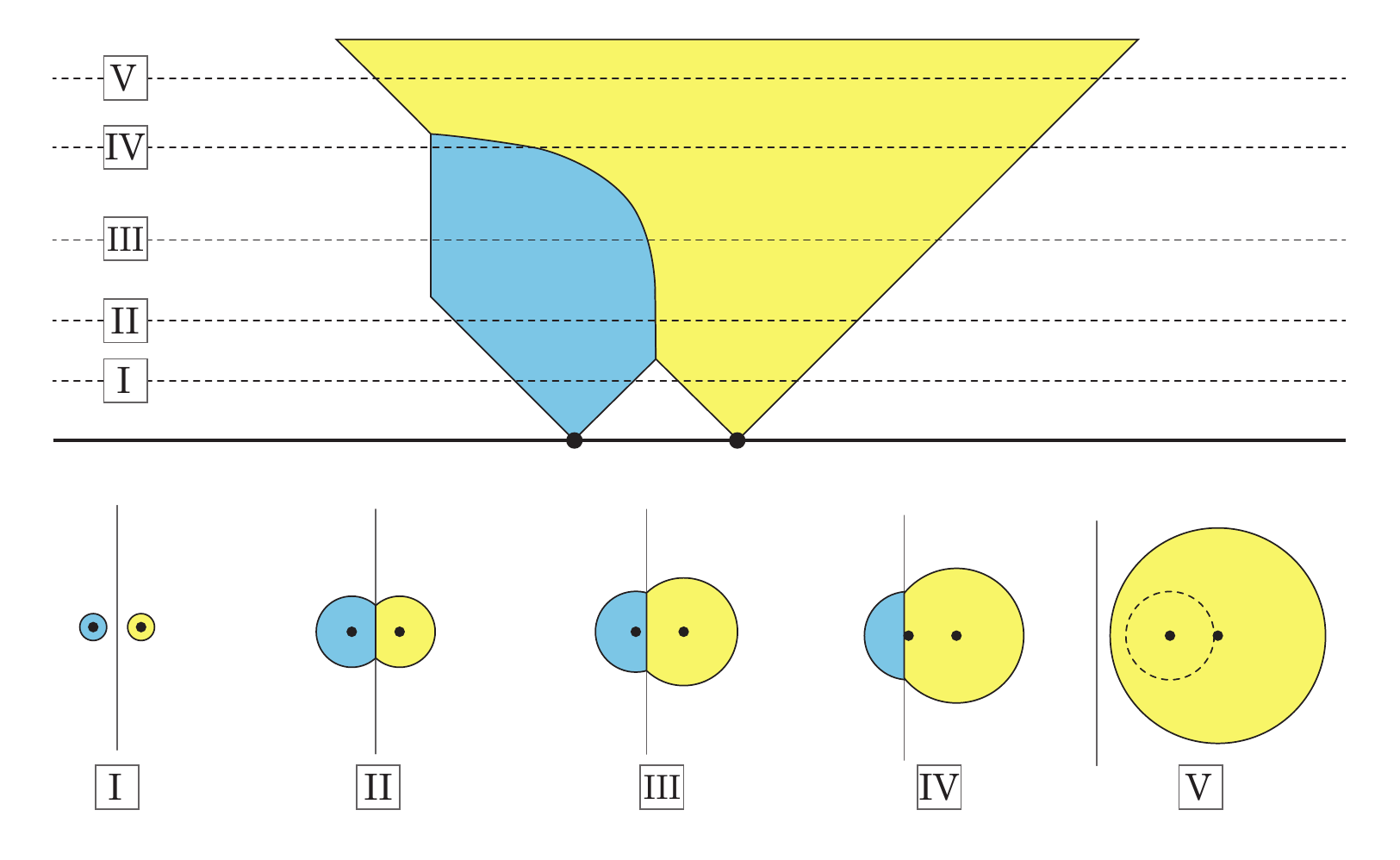}
  \caption{For two points, the lifted, clipped Voronoi cells are shown from the side.  The clipped Voronoi diagrams at five different scales are illustrated.  Note that at some scale, the left ball ceases growing.  Then, it is overtaken by the cell of the right point.}\label{fig:lifted_clipped_cells}
\end{figure}

\begin{restatable}{lemma}{emptyclippedcells}\label{lem:empty_clipped_cells}
  Let $P\subset \R^d$ be ordered according to a greedy permutation with insertion radii $r_0 \ldots r_{i-1}$.
  Let $\e>0$ be any constant.
  For all $j\in [1:n]$, let the freezing times $\lambda_j$ be chosen so that
  \[
    \lambda_j \ge \frac{1+\e}{\e}r_j.
  \]
  Then, for all $k\in [1:n]$ and all $\alpha> (1+\e)\lambda_k$, we have $V_k^\alpha= \emptyset$.
\end{restatable}

This is how some points will cease to impact the filtration at larger scales.
In the next section, we will see how to simulate the removal of vertices whose clipped Voronoi cells are empty.
Before we give that construction, we will relate the clipped Voronoi diagram to the (weighted) offsets.
For completeness, we include a proof of this well-known fact (see for example Edelsbrunner~\cite{edelsbrunner95union}) in Appendix~\ref{sec:covering_lemmas}.

\begin{restatable}{lemma}{clippedcellscoveroffsets}\label{lem:clipped_cells_cover_offsets}
  Let $P\subset \R^d$ be ordered according to a greedy permutation with insertion radii $r_0 \ldots r_{i-1}$.
  Let $\e>0$ be any constant.
  For all $j\in [1:n]$, let the freezing times $\lambda_j$ be chosen so that
  \[
    \lambda_j \ge \frac{1+\e}{\e}r_j.
  \]
  Then, for all $\alpha\ge 0$,
  \[
    \hat{P}^\alpha = \bigcup{i\in [n]}V_i^\alpha.
  \]
\end{restatable}

\subsection{Extending Voronoi and Delaunauy in Time and Space}

Given that we will be considering weighted point sets in which the weights of the points vary in time, it is useful to give a concrete geometric structure that captures the evolution of the Voronoi diagram and the Delaunay triangulation.

Define
\[
  \hat{P}_+^\alpha := \{(x,\gamma) \in \R^d \times [0,\alpha] \mid x\in \hat{P}^\alpha\}
  .
\]
This filtration has the property that for any $\alpha\ge 0$, there is a natural homotopy equivalence $m:\hat{P}_+^\alpha \to \hat{P}^\alpha$ defined by the projection $m(x,\gamma) = x$.
This follows from the fact that $\hat{P}\alpha$ only grows with increasing $\alpha$ and thus the fibers of $m$ are simply connected (line segments).

% TODO: commutative diagram

We can similarly embed the clipped Voronoi cells in $\R^{d+1}$ by defining
\[
  V_{+i}^\alpha := \{(x,\gamma)\in \R^d\times [0,\alpha] \mid x\in V_i^\alpha\}.
\]
The collection of these sets is
\[
  V_+^\alpha := \{V_{+i}^\alpha\}_{i\in [n]}.
\]

Let $D_+^\alpha$ be the union of the Delaunay complexes at time $\gamma$ for all $\gamma\in [0,\alpha]$.
% , because every intersection of cells in $\R^{d+1}$ is an intersection in some slice.

\subsection{Nerves}

Construction of the Delaunay triangulation from the Voronoi cells is an example of a \emph{nerve}.
More generally, given a collection $U$ of sets, we define a simplicial complex
\[
  \nerve(U) := \{\sigma\subset U \mid \bigcap_{S\in \sigma} S \neq \emptyset\}.
\]
For the Delaunay triangulation, if there is a nonempty intersection of Voronoi cells, we identify the corresponding simplex with the set of points defining those cells.
So, $D^\alpha$ is (isomorphic to) the nerve of $V^\alpha$ and $D_+^\alpha$ is the nerve of $V_+^\alpha$.

The collection $U$ is called a \emph{cover}, and the set $W = \bigcup_{S\in U} S$ is the set that $U$ covers.
A cover is a \emph{good cover} if the intersections of elements are all either empty or contractible.
In the case of the Voronoi diagram as well as the clipped Voronoi diagram, the convexity of the cells guarantees that the cover is good.

The \emph{Nerve Theorem} says that the nerve of a good cover is homotopy equivalent to the union, so the homology of the nerve matches the homology of the union.
In the case of the clipped Voronoi diagram, the Nerve Theorem implies that the Delaunay complex at scale $\alpha$ is equivalent in homology to the weighted offsets.
This observation was one of the main ideas that drove the development of persistent homology.

If, instead of a collection of sets, we have a collection of filtrations, we can define their nerve as a filtered simplicial complex.
For example, if we have filtrations $\{(U_0^\alpha), \ldots, (U_{k-1}^\alpha)\}$, then for each $\alpha$ we can define the set $W^\alpha = \bigcup_{i\in [k]}U_i^\alpha$ and the nerve $N^\alpha = \nerve(\{U_i^\alpha \mid i \in [k]\})$.
Then, we have a new filtration $(W^\alpha)$ and the filtered simplicial complex $(N^\alpha)$.

The Persistent Nerve Lemma~\cite{chazal08towards} implies that if $\{U_i^\alpha \mid i \in [k]\}$ is a good cover for all $\alpha$, then the persistence diagrams of $(N^\alpha)$ and $(W^\alpha)$ are identical.
% For the weighted offsets and their cover by clipped Voronoi cells, this means that $\dgm(\hat{P}^\alpha) = \dgm(D^\alpha)$.

% \subsection{Sparse Filtrations}
%
% When computing the persistent homology of a distance function, the main limitation is space rather than time.
% For an input of size $n$ in $d$ dimensions, it is not uncommon to produce a complex of size $\Theta(n^{d+1})$ as input to the persistence algorithm.

  \section{The Sparse Delaunay Filtration}
\label{sec:correctness}

In this section, we will prove that our sparse Delaunay filtration $\dgm(D_+^\alpha)$ is a good approximation to $\dgm(P^\alpha)$.
%
% The relationship between $P^\alpha$ and $D^\alpha$ does not depend on the specific choice of weights.
% However, as Lemma~\ref{lem:empty_clipped_cells} shows, the sequence $(D^\alpha)$ is not necessarily a filtration when the weights change in time.
% The remedy for this is to consider instead the filtration $(D_+^\alpha)$.
The high-level plan for the analysis is to relate $\dgm(P^\alpha)$ to $\dgm(\hat{P}^\alpha)$ to  $\dgm(\bigcup V_+^\alpha)$ to $\dgm(D_+^\alpha)$.
% In order for this plan to succeed, one must be careful about the choice of weights.
We have already shown the first two steps follow from standard geometric arguments.
For the last step, we will use the Persistent Nerve Lemma, but it will require showing that lifted clipped Voronoi cells form a good cover.
Although the slices of these cells are convex at any fixed $\alpha$, they are not themselves convex as can be seen in the example in Fig.~\ref{fig:lifted_clipped_cells}.
So, proving its a good filtered cover will depend on our careful choice of weights.
% So, it is not immediately clear that they form a good filtered cover, and, for some weight functions, they will not be.

% In this section, we will give the specific choice of weights that we will use in our construction and prove that $\{V_i^\alpha \mid i \in [n]\}$ is a good cover of $\hat{P}^\alpha$.

\subsection{Adding Points in Waves}

To define the weight functions, it suffices to establish the freezing times.
The specific choice of the freezing times impacts the size (larger freezing times yields a larger filtration), but also the correctness (the cover must be good).
For many of the preceding lemmas, it was necessary to choose freezing times so that
\[
  \lambda_i \ge \frac{1+\e}{\e}r_i,
\]
where $r_i$ s the insertion radius of $p_i$ in a greedy ordering of the input set $P$.
Recall that $\e$ is the user chosen parameter that will define the accuracy of our approximation.

We will satisfy this requirement by grouping points according to their insertion radius, rounding to the nearest power of $1+\e$.
All the points in a group will have the same freezing time.
Formally, we set
\[
  \lambda_i := (1+\e)^{\left\lceil \log_{1+\e} \left(\frac{1+\e}{\e}r_i\right)\right\rceil}.
\]

The filtration $(D_+^\alpha)$ defined with these weight functions on a greedy permutation is \emph{The Sparse Delaunay Filtration}.
% We will say that these freezing times are

\subsection{Monotonicity}

The first step in showing that $V_+^\alpha$ is a good filtered cover is to show that every Delaunay simplex in $D^\alpha$ appears and disappears at most once.
That is, if we watch the evolution of the $d$-dimensional weighted Delaunay complex $D^\alpha$ as $\alpha$ increases, then no simplex that leaves will ever come back.

There are two ways that a simplex can be removed at time $\alpha$.
First, it may be that its Voronoi cell becomes empty and therefore it is removed from the Delaunay triangulation entirely.
This is the standard case to analyze in flip-based Delaunay computation.
Second, it may be that only the clipped Voronoi cell becomes empty.
In this case, we must show that it remains empty for the rest of the filtration.

The challenge is that the lifted Voronoi cells are not convex.
Aurenhammer et al.~\cite{aurenhammer17voronoi} showed that this monotonicity does not hold for a related set of weight functions.
That paper addressed the problem of computing the Voronoi diagram of parallel halflines and showed that the slices perpendicular to the halflines are weighted Voronoi diagrams.
Similar to the current paper, they construct a $d+1$-dimensional decomposition by sweeping through $d$-dimensional slices.
That paper gives an example attributed to Peter Widmayer's research group in which a particular triangle would be flipped out and later flipped back in.
Such a non-monotone example highlights the need to be careful in choosing the weights.

\begin{lemma}\label{lem:monotonicity}
  If $\sigma\in D^\alpha$ and $\sigma\in D^\beta$, then $\sigma\in D^\gamma$ for all $\gamma\in [\alpha,\beta]$.
\end{lemma}
\begin{proof}
  Let $\lambda_\sigma =  \max_{p_i\in \sigma}\lambda_i$.
  If $\beta\le\lambda_\sigma$, then $\bigcap_{p_i\in \sigma}V_i^\alpha \subseteq \bigcap_{p_i\in \sigma}V_i^\gamma$ for all $\gamma\in [\alpha, \beta]$, because the points of $\sigma$ all have weight $0$ in this range.
  Lemma~\ref{lem:empty_clipped_cells} implies that no point $p_i$ appears in $D^\alpha$ for $\alpha\ge (1+\e)\lambda_i$.
  So, it will suffice to prove the lemma for the case where $\alpha$ and $\beta$ are in the interval $[\lambda_\sigma, (1+\e)\lambda_\sigma]$.

  For any $\gamma\in [\alpha,\beta]$, let
  \[
    t = \frac{\gamma^ - \alpha^2}{\beta^2 - \alpha^2}.
  \]
  This choice implies that
  \[
    \gamma = \sqrt{(1-t)\alpha^2 + t \beta^2}.
  \]
  Moreover, for all $p_i\in P$, we have
  \[
    w_i(\gamma)^2 = (1-t)w_i(\alpha)^2 + tw_i(\beta)^2.
  \]
  % Two cases: if p_i is unfrozon, then it follows from the weights being zero.
  % if p_i is frozen, then it's just a little algebra.

  Let $x$ be a point in the intersection $\bigcap_{p_i\in \sigma} V_i^\alpha$ and let $y$ be a point in $\bigcap_{p_i\in \sigma} V_i^\beta$.
  These points witness the existence of $\sigma$ in $D^\alpha$ and $D^\beta$ respectively.
  We will show that $z = (1-t)x + ty$ is in $\bigcap_{p_i\in \sigma} V_i^\gamma$.

  Using the convexity of the squared distance, we observe that for all $p_i\in \sigma$
  \begin{align*}
    \pdist_{i, \gamma}(z)^2
      &= \|p_i - z\|^2 + w_i(\gamma)^2\\
      &= \|p_i - z\|^2 + (1-t)w_i(\alpha)^2 + tw_i(\beta)^2\\
      &\le (1-t)\|p_i-x\|^2 + t\|p_i - y\| + (1-t)w_i(\alpha)^2 + tw_i(\beta)^2\\
      &= (1-t)\pdist_{i,\alpha}(x) + t\pdist_{i_\beta}(y)\\
      &\le (1-t)\alpha^2 + t\beta^2\\
      &= \gamma.
  \end{align*}
  So, $z\in b_i(\gamma)$ for all $p_i\in \sigma$.

  Next, we show that $z\in \vor_P^\gamma(p_i)$ for all $p_i\in \sigma$.
  By the definition of the power distance,
  \[
    \pdist_{i,\alpha}(x) \le \pdist_{j,\alpha}(x) \text{ iff } 2x^\top(p_j-p_i)\le w_j(\alpha)^2 - w_i(\alpha)^2.
  \]
  Similarly,
  \[
    \pdist_{i,\beta}(y) \le \pdist_{j,\beta}(y) \text{ iff } 2y^\top(p_j-p_i)\le w_j(\beta)^2 - w_i(\beta)^2.
  \]
  So, for all $p_i\in \sigma$ and all $p_j\in P$,
  \begin{align*}
    2z^\top (p_j - p_i)
      &= (1-t)2x^\top (p_j - p_i) + t2y^\top(p_j - p_i)\\
      &\le (1-t)(w_j(\alpha)^2 - w_i(\alpha^2)) + t(w_j(\beta)^2 - w_i(\beta^2))\\
      &= w_j(\gamma)^2 - w_i(\gamma)^2.
  \end{align*}
  So, for all such $p_i,p_j$, we have $\pdist_{i,\gamma}(z)\le \pdist_{j,\gamma}(z)$ and thus $z\in \vor_P^\gamma(p_i)$.

  As we have shown $z\in b_i^\gamma$ and $z\in \vor_P^\gamma(p_i)$, it follows that $z\in b_i^\gamma \cap \vor_P^\gamma(p_i) = V_i^\gamma$ for all $p_i\in \Sigma$ and $z\in \bigcap_{p_i\in \sigma}V_i^\gamma$.
  Therefore, $\sigma\in D^\gamma$ as desired.
\end{proof}

\subsection{A Good Filtered Cover}

We can now prove that the lifted Voronoi cells form a good filtered cover.
% Then, we use this to show that for our specific choice of weights, the

\begin{lemma}\label{lem:good_cover}
  Let $\e>0$ be any constant.
  Let $V_+^\alpha$ be the lifted Voronoi cells whose nerve is the $\e$-Sparse Delaunay Filtration for $P\subset \R^d$.
  Then, $(V_+^\alpha)$ is a good filtered cover of $(\hat{P}^\alpha)$.
\end{lemma}
\begin{proof}
  Fix any $\alpha\ge 0$.
  Let $\sigma \in \nerve(V_+^\alpha)$ be any simplex.
  We will show that the intersection of the lifted, clipped Voronoi cells $\{V_i^\alpha \mid p_i\in \sigma\}$ is contractible.

  In each $d$-dimensional slice, the clipped Voronoi cells $V_i^\alpha$ are convex, so their intersection is convex.
  We can deformation retract the nonempty clipped Voronoi cells in each slice to the orthocenter, i.e., the point in the cell that minimizes the distance to the points of $\sigma$.
  The cells change continuously in time and so does the orthocenter of the simplex.
  The retractions in each slice will be continuous as a retraction in $\R^d+1$.
  By Lemma~\ref{lem:monotonicity} the collection of orthocenters in the slices will form a connected path and thus are contractible.
\end{proof}

\begin{theorem}\label{thm:correctness}
  Let $\e>0$ be any constant.
  Let $(D_+^\alpha)$ be the $\e$-Sparse Delaunay Filtration for $P\subset \R^d$.
  Then, $\dgm(D_+^\alpha)$ is a $(1+\e)$-approximation to $\dgm(P^\alpha)$.
\end{theorem}

\begin{proof}
  By Lemma~\ref{lem:good_cover}, the lifted clipped Voronoi cells $\{V_i^\alpha\}$ form a good cover of $\hat{P}^\alpha$ for all $\alpha$ and therefore, by the Persistent Nerve Lemma,
  \[
    \dgm(D^\alpha) = \dgm(\hat{P}^\alpha).
  \]
  Lemma~\ref{lem:weighted_offset_interleaving} gives an interleaving of persistence modules (see~\cite{chazal16structure}) so that $\dgm(\hat{P}^\alpha)$ is a $(1+\e)$-approximation to $\dgm(P^\alpha)$.
  Combining these facts, we get that $\dgm(D^\alpha)$ is a $(1+\e)$-approximation to $\dgm(P^\alpha)$.
\end{proof}

  \subsection{Size Analysis}

% The main goal of The Sparse Delaunay Filtration is to reduce the size of the filtered complex used in computing the persistent homology of offsets.
In general, the Delaunay complex on $n$ points (in general position) may have $O(n^{\lceil d/2 \rceil})$ simplices~\cite{seidel87number}.
There are several special cases where it is known that the Delaunay complex has size $O(n)$.
Most such cases are based on input models that guarantee the points are spaced according to some Poisson process.
The analysis invariably depends on showing that the complex is everywhere locally sparse in the sense that every vertex participates in at most a constant number of simplices.
It will not be hard to show that a similar bound holds for every slice and, in particular, every slice has linear size (Lemma~\ref{lem:slices_are_sparse}).
Then, we will show that the total number of simplices (the union of all slices) also has linear size (Theorem~\ref{thm:linear_size}).

\begin{lemma}\label{lem:slices_are_sparse}
  For all $\alpha\ge 0$, every vertex in the weighted Delaunay complex $D^\alpha$ has at most $O\left(\left(\frac{1+\e}{\e}\right)^d\right)$ neighbors.
\end{lemma}
\begin{proof}
  Let $p$ be any point in $P$.
  Let $Q = \{q_0,\ldots, q_{k-1}\}$ be the neighbors of $p$ in $D^\alpha$.
  Every edge of $D^\alpha$ has length at most $2\alpha$.
  By Lemma~\ref{lem:empty_clipped_cells}, every $q_i$ must have a freezing time at least $\frac{\alpha}{1 + \e}$ and thus, an insertion radius of at least $\frac{\e\alpha}{(1+\e)^3}$.
  So, the points of $Q$ are contained in $\ball(p,2\alpha)$ and are all pairwise $\frac{\e\alpha}{(1+\e)^3}$-separated.
  By comparing the volumes of the $k$ disjoint empty balls of radius $\frac{\e\alpha}{2(1+\e)^3}$ around the points of $Q$ to the volume of the ball of radius that $\left(2+ \frac{\e}{(1+\e)^3}\right)\alpha$ that contains them, we get that $k = O\left(\left(\frac{1+\e}{\e}\right)^d\right)$.
\end{proof}

\begin{theorem}\label{thm:linear_size}
  Let $P$ be a set of $n$ points in $\R^d$ with greedy weights.
  Let $\e\ge 0$ be a constant.
  The total size of $(D^\alpha)$ is $O(n)$.
\end{theorem}

\begin{proof}
  The proof follows the exact pattern of previous work on sparse filtrations (see \cite{sheehy13linear,cavanna15geometric} for more a more detailed analysis).
  Each simplex of the filtration is charged to the vertex with the smallest insertion radius.
  By a packing volume argument analogous to that in Lemma~\ref{lem:slices_are_sparse}, we see that no vertex is charged for more than a constant number of simplices in the final filtration.
  Thus the total size is $O(n)$ as desired.
\end{proof}

  \section{Efficient Construction}
\label{sec:algorithm}

The main challenge to efficiently constructing the Sparse Delaunay Filtration is to avoid constructing the entire Delaunay complex.
Doing so could easily negate any efficiency gains from sparsification.
In this section, we will describe an approach based on Voronoi refinement that uses extra points called Steiner points to keep the complexity of the Delaunay complex linear in the number of points.
None of the Steiner points will appear in the output filtration.
They serve only to fill in large gaps that could potentially create a superlinear number of Delaunay simplices.

Earlier work in approximating the persistence diagram of the distance to Euclidean points also used Steiner points~\cite{hudson10topological,sheehy11thesis}, but in that case, the Steiner points were an essential part of the filtration.
Steiner points and Voronoi refinement have also been used in output-sensitive algorithms to construct $D^\alpha$~\cite{sheehy15output}.
In that case, the output could still be superlinear in the input size depending on the arrangement of the points and the choice of $\alpha$.

Flip-based algorithms for computing the Delaunay triangulation start the insertion of a new point by flipping a new vertex in with a $(1, d+1)$-flip.
This requires that the new point is contained in one of the simplices of the current triangulation.
If we ignore or discard simplices in the Delaunay triangulation that do not appear in the subcomplex $D^\alpha$, then we cannot necessarily flip in new points, because we could have discarded a simplex containing the new point.
If we maintain the full Delaunay triangulation at all times, we might store too many simplices.
To balance between the two, we use Steiner point to give a sparse representation of the regions far from the input points at a given scale.
In Section~\ref{sec:voronoi_refinement}, we explain how these Steiner points are chosen.
Then, in Section~\ref{sec:no_steiners}, we show why these Steiner points do not affect the output.
The full algorithm is then presented in Section~\ref{sec:full_algorithm} and analyzed in Section~\ref{sec:analysis}

\subsection{Voronoi Refinement}
\label{sec:voronoi_refinement}

Voronoi cells are polyhedra whose vertices we will call \emph{corners} to distinguish them from the input vertices.
We will assume that the affine closure of the points is $d$-dimensional, so every Voronoi cell has at least one corner.
Let $P\subset\R^d$ and let $y\in P$ be any point.
Let $z$ be the nearest neighbor of $y$ in $P$, and let $x$ be the corner of $\vor_P(y)$ farthest from $y$.
The \emph{aspect} of $\vor_P(y)$ is
\[
  \aspect(y) := \frac{\|y-x\|}{\|y-z\|}.
\]
The point set $P$ is \emph{$\tau$-well-spaced} if $\aspect(y) <\tau$ for all $y\in P$.
There are many advantages to well-spaced points when constructing Delaunay triangulations.
A major advantage is that the number of simplices incident to any vertex will be at most a constant.
So, the total complexity of the Delaunay triangulation of $n$ well-spaced points is at most $O(n)$.

Voronoi refinement is a variant of Delaunay refinement~\cite{chew93guaranteed,ruppert95delaunay} and is commonly used in mesh generation (see the books by Edelsbrunner~\cite{edelsbrunner01geometry} and Cheng et al.~\cite{cheng12delaunay} for more details).
The corners of a Voronoi cell are the circumcenters of their dual Delaunay simplices.
The basic Voronoi refinement algorithm is to add the farthest corner of any cell $\vor_P(y)$ for which $\aspect(y)>\tau$.
Repeating this process eventually produces a $\tau$-well-spaced set of points.
Moreover, the total number of points in the output is asymptotically optimal~\cite{ruppert95delaunay,sheehy12new}, i.e., the size is within a constant factor of any $\tau$-well-spaced superset of $P$.

\subsection{Why the Steiner points don't appear in the filtration}
\label{sec:no_steiners}

Every edge in $D^\alpha$ is induced by the intersection of two clipped Voronoi cells, so the length of every edge is at most $2\alpha$.
The following lemma is the key to guarantee that the Sparse Delaunay Filtration we construct contains no Steiner points as vertices.
It shows that the Steiner points are always more than $2\alpha$ away from any other points, and therefore, there can be no edges incident to a Steiner point in $D^\alpha$.

\begin{lemma}\label{lem:steiner_points_stay_far_from_P}
  Let $\e$ and $\alpha$ be nonnegative real numbers.
  Let $P$ and $S$ be subsets of $\R^d$ such that no point of $S$ is within $2\alpha$ and no point of $P$ is within $\e\alpha$ of any other point of $P\cup S$.
  If $S'$ is formed by adding Steiner points to $S$ at the far corners of Voronoi cells with aspect greater than $\frac{2}{\e}$, then no point of $S'$ will be within $2\alpha$ of any other point of $P\cup S'$.
\end{lemma}
\begin{proof}
  It suffices to show that the spacing condition holds after the insertion of each Steiner point $x$.
  Let $y$ be the point whose Voronoi cell was refined by the addition of $x$ and let $z$ be the nearest neighbor of $y$.
  Then, $\|y-z\|\ge \e\alpha$ and the aspect of the cell is bounded as
  \[
    \frac{2}{\e} < \frac{\|y-x\|}{\|y-z\|} \le \frac{\|y-x\|}{\e\alpha}.
  \]
  It follows that $\|y-x\| > 2\alpha$.  Because $x$ was in the Voronoi cell of $y$, it follows that the distance to any other point is also greater than $2\alpha$.
\end{proof}

\subsection{The Full Algorithm}
\label{sec:full_algorithm}

In the preprocessing phase of the algorithm, we compute a greedy permutation of $P$.
During this computation, we also compute for each $p_i$, the nearest predecessor in the ordering as well as its distance, the insertion radius $r_i$.
The radius will be used to establish the weights and define the waves.
The nearest predecessor will be used for point location when inserting new points.

The algorithm then proceeds by constructing the filtration one wave at a time in order of the greedy permutation.
That is, for wave $w$, the filtration is constructed for the interval $[(1+\e)^w,(1+\e)^{w+1}]$.
Moreover, as the construction increases the density of points as it goes, it also decreases the radius, so the simplices of the filtration are discovered in reverse order.

At the start of each wave, some set $U$ of points have already been inserted, and some set $F$ of points will be inserted into to the Delaunay triangulation.
The points $U$ are unfrozen throughout the wave so their weights will always be zero.
The points $F$ are frozen at the start of the wave interval, so their weights will vary equally in time.
We start the wave by locating the Delaunay simplices containing each of the points of $F$.
Each pair of a point and the $d$-simplex that contains it forms a $(d+1)$-simplex.
We compute the flip time for each of these simplices and store the flip in the event queue.

Processing the flips only requires that we remove the next flip from the event queue (i.e., the maximum flip time).
We check that the simplices are still present in the complex, i.e., that no other flip removed some of its subsimplices.
Then, we execute the flip, updating the Delaunay triangulation, and relocating uninserted points of $F$ that were in the removed simplices.

For each flip, we update the birth times of all simplices that may have been affected.
That is, if $S$ is the set of $d+2$ points involved in the flip at scale $\alpha$, we will process each simplex $\sigma\subseteq S$ starting with the highest dimensions.
The tentative birth time of every simplex is computed assuming that the structure of the triangulation will not change.
A simplex is \emph{finalized} when we add it to the filtration.
A simplex is \emph{discarded} if we have established that it will never appear in the filtration.
Discarded simplices do not need to be updated in this step.
For $\sigma = S$, set $\birth(S) = \alpha$ and either finalize it if $\radius(S)\le \alpha$, or discard it otherwise.
For simplices $\sigma$ removed by the flip, either finalize it if $\birth(\sigma) \ge \alpha$ or discard it otherwise.
For all other simplices, if the newly computed birth time is at most $\alpha$, then update $\birth(\sigma)$ and finalize $\sigma$ otherwise.
% Every simplex $\sigma$ that appears in the course of the construction, including the $d+1$-simplices corresponding to flips has a time $\beta$ when it is flipped in and (possibly) a time $\alpha$ when it is flipped out.
% If $\sigma$ is to appear in $D^\gamma$, then $\gamma\in [\alpha, \beta]$.
% In the standard Delaunay filtration, the birth time of a simplex is either its circumradius if the circumsphere is empty of other points or the minimum birth time of any simplex that contains it.
% The same will hold for our complex, except it will be more complicated due to the changing weights.

% Claim: the orthocenter stays on the straight line between the orthcenter of the flip where it first appear and the flip where it disappears.
% The $(d+1)$-simplex associated with the flip is added to the filtration with birth time equal to its flip time.
% TODO: what are the birth times of the simplices?

At the end of each wave, we perform a Voronoi refinement step, adding Steiner points until the points are $\frac{2}{\e}$-well-spaced.
That is, while any Voronoi cell has aspect greater than $\frac{2}{\e}$, we add it farthest corner.
As a consequence of Lemma~\ref{lem:steiner_points_stay_far_from_P}, no changes to the filtration are made at this time.

The overall running time is just the cost of computing a greedy permutation, doing an incremental Delaunay triangulation with sparse refinement, and performing a constant amount of extra work per flip.
These are relatively standard analyses, so in the interest of space, they have been relegated to Appendix~\ref{sec:analysis}.

  \section{Conclusion}
\label{sec:conclusion}

We have presented a linear size Delaunay filtration for $n$ points in $\R^d$ as well an efficient algorithm to compute it.

It is also relevant to note that this algorithm also can be used to compute a well-spaced set of points.
That is, if one keeps the final Delaunay triangulation of the input plus the Steiner points, the result will be well-spaced.
Performing a more aggressive Voronoi refinement to achieve a better spacing constant then resembles a standard Voronoi/Delaunay refinement starting from a well-spaced point set.
This is substantially simpler than previous algorithms to do Sparse Voronoi Refinement~\cite{hudson06sparse,miller11beating} because it obviates any need to ``snap'' Steiner points to nearby input points.
It is not obvious  whether the tradeoff between a size increase from the difference in the spacing constant offsets the improvements from simplified point location.

  \bibliographystyle{abbrv}
  \bibliography{../../bib/bibliography}

  \appendix
  \section{Covering Lemmas}
\label{sec:covering_lemmas}

In this appendix, we give the full proofs of the lemmas that are used for analyzing the weighted offset filtration and the corresponding Voronoi filtration.

\weightedcover*

\begin{proof}
  If $\lambda_k \ge (1+\e)\alpha$, then choosing $i=k$ suffices to satisfy the two conditions.
  Otherwise, let $j$ be the maximum such that $\lambda_j \ge (1+\e)\alpha$.
  So, $\lambda_{j+1} < (1+\e)\alpha$.
  Choose $i\le j$ to minimize $\|p_i - p_k\|$.
  Then, $\lambda_i \ge \lambda_j\ge (1+\e)\alpha$ as desired.
  Using the covering property of the greedy permutation,
  \[
    \|p_i - p_k\| \le r_{j+1} \le \frac{\e}{1+\e}\lambda_{j+1} <\e \alpha.
  \]
\end{proof}

\weightedoffsetinterleaving*

\begin{proof}
  First, we will show that $\hat{P}^{\alpha} \subseteq P^\alpha$.
  It suffices to observe that the introduction of nonnegative weights can only increase the distance to the set.
  By definition, any point in $\hat{P}^\alpha$ has a point in $\hat{P}(\alpha)$ within weighted distance $\alpha$.
  The corresponding unweighted point in $P$ is within distance $\alpha$, so $x\in P^\alpha$ as desired.

  Next, we show that $\subseteq P^\alpha \subseteq \hat{P}^{(1+\e)\alpha}$.
  For any $x\in P\alpha$, there exists an index $k$ such that $\|p_k - x\| \le \alpha$.
  By Lemma~\ref{lem:weighted_cover}, there exists $i\le k$ such that $\lambda_i \ge (1+\e)\alpha$, and $\|p_i - p_k\| \le \e \alpha$.
  So, $w_i((1+\e)\alpha) = 0$ and thus
  \begin{align*}
    \pdist_{\hat{p_i}, \alpha(1+\e)}(x)
      &= \sqrt{\|p_i - x\|^2 + w_i(\alpha)^2}\\
      &= \|p_i - x\|\\
      &\le \|p_i - p_k\| + \|p_k - x\|\\
      &\le \e \alpha + \alpha\\
      &= (1+\e)\alpha
  \end{align*}
  Therefore, $x\in \hat{P}^{(1+\e)\alpha}$, and so, $\subseteq P^\alpha \subseteq \hat{P}^{(1+\e)\alpha}$.
\end{proof}

\emptyclippedcells*

\begin{proof}
  Choose any $k\in [1:n]$ and any value of $\alpha > (1+\e)\lambda_k$.
  Suppose for contradiction that there is a point $x\in V_k^\alpha$.
  In particular, $x\in b_k(\alpha)$ and so $\|p_k - x \|\le \lambda_k$.
  By Lemma~\ref{lem:weighted_cover}, there is a point $p_i$ such that $\|p_i - p_k\| < \e \lambda_k$ and $\lambda_i \ge (1+\e)\lambda_k$.
  % We have supposed $x\in V_i^\alpha$.
  So, $w_i(\alpha) = 0$ and thus,
  \begin{align*}
    \pdist_{\hat{p_i}^\alpha}(x)^2
      &= \|p_i-x\|^2 + \alpha^2 - \lambda_i^2 \\
      &\le (\|p_i - p_k\| + \|p_k - x\|)^2 + \alpha^2 - \lambda_i^2\\
      &< (\e\lambda_k + \|p_k - x\|)^2 + \alpha^2 - \lambda_i^2\\
      &= \e^2\lambda_k^2 + 2\e\lambda_k\|p_k - x\| + \|p_k - x\|^2 + \alpha^2 - \lambda_i^2\\
      &\le (\e^2 + 2\e)\lambda_k^2 + \|p_k - x\|^2 + \alpha^2 - \lambda_i^2\\
      &\le (\e^2 + 2\e)\lambda_k^2 + \|p_k - x\|^2 + \alpha^2 - (1+\e)\lambda_k^2\\
      &= \|p_k - x\|^2 + \alpha^2 - \lambda_k^2\\
      &= \pdist_{\hat{p_k}^\alpha}(x)^2.
  \end{align*}
  It follows that $x\notin \vor_{\hat{P}}^\alpha(\hat{p}_k)$ and therefore $x\notin V_i^\alpha$, a contradiction.
\end{proof}

\clippedcellscoveroffsets*

\begin{proof}
  By definition, every clipped cell $V_i^\alpha$ is contained the ball $b_i^\alpha$ so
  \[
    \bigcup_{i\in [n]}V_i^\alpha \subseteq \bigcup_{i\in [n]}b_i^\alpha = \hat{P}^\alpha.
  \]

  Next, observe that for all $x\in \hat{P}^\alpha$, there is a nearest weighted point $\hat{p}_i\in \hat{P}$.
  In other words, $x\in \vor_{\hat{P}, \alpha}(\hat{p}_i)$.
  By the definition of the offsets, $\pdist_{\hat{p}_i, \alpha} \le \alpha$.
  Therefore, $x\in b_i^\alpha$.
  So, $x \in \vor_{\hat{P}, \alpha}(\hat{p}_i) \cap b_i^\alpha = V_i^\alpha$.
  It follows that $\hat{P}^\alpha \subseteq \bigcup_{i\in [n]}V_i^\alpha$.
\end{proof}

  \section{Analysis}
\label{sec:analysis}

  There are four main parts of the computation to analyze:
  \begin{enumerate}
    \item the preprocessing phase, which includes the computation of the greedy permutation, the finding of nearest predecessors, and the establishing of freezing times and weight functions;
    \item the point location work to begin the insertion of new points into the triangulation;
    \item the global ordering of flips required during the computation as weights change; and
    \item the cost of Voronoi refinement to maintain the aspect ratio bound.
  \end{enumerate}

\paragraph*{Preprocessing}
The straightforward algorithm for computing the greedy permutation as presented by Gonzalez~\cite{gonzalez85clustering} as well as Dyer and Frieze~\cite{dyer85simple} runs in quadratic time.
An approach that exploits the intrinsic low-dimensionality to reduce the number of distance computations was developed by Clarkson~\cite{clarkson03nearest} and was shown to run in $O(n\log\spread)$ time by Har-Peled and Mendel~\cite{har-peled06fast}.
The constants here and throughout depend exponentially on the dimension.
Variants of these algorithms are available in the Python \texttt{greedypermutations} library~\cite{sheehy20greedypermutations}.
Har-Peled and Mendel also showed that a significantly more complex algorithm can produce a greedy permutation in $O(n\log n)$ time.
All of the existing algorithms also produce, as a byproduct, the nearest predecessors and thus, the insertion radii.

\paragraph*{Point Location}
The first step to insert a new point into a Delaunay triangulation is to identify the triangle that contains the point; this is the \emph{point location} step.
For randomized incremental construction, this can be accomplished by maintaining the location of each uninserted point and updating it as necessary with each flip (see, for example Guibas et al.~\cite{guibas92randomized}).
Another approach is to find each point by walking from triangle to triangle and this is often the fastest in practice~\cite{devillers98improved}

More generally, our algorithm resembles the Sparse Voronoi Refinement algorithm of Hudson et al.~\cite{hudson06sparse} in that it alternates insertion of input points with refinement.
The point location in that algorithm works by associating each uninserted point with the set of Delaunay circumspheres that contain it.
The total cost of point location in that algorithm is $O(n\log \spread)$.
It was later shown that it was more efficient to store the uninserted points in the Voronoi cells that contain them~\cite{acar07svr}.
Miller et al.~\cite{miller13fast} showed that for well-spaced points, a variation of the walk-based point location that identifies the containing Voronoi cell can be performed in constant time per point if the points are added in a greedy order and the walk starts from the nearest predecessor in the ordering.
This approach applies directly to our algorithm as well, and is especially convenient given that we needed to compute the greedy permutation anyway.
Most of the point location work has been shifted to the greedy permutation and the total cost of the point location walks is $O(n)$.

\paragraph*{Tracking Flips}
A priority queue is used to track the flips that can or will occur in the course of the construction.
The flips associated with the insertion of Steiner points are not tracked in this priority queue.

Within any given wave, every flip adds an edge to one of the vertices that is added in that wave.
The degree of any vertex that it is added in the wave is at most a constant, because the points are well-spaced.
This means that the number of potential flips (not all flips are executed) is $O(1)$ per points.
As a result, there are $O(n)$ total flips in the event queue
So, the total work of maintaining thee event queue is $O(n\log n)$.

The number of flips in a wave can be counted precisely using the observation of Miller and Sheehy~\cite{miller14new} in their work on an output-sensitive algorithm for computing Voronoi diagrams.
They showed that the removal of a subset of points by adjusting the weights in lockstep, which is equivalent to the action of one wave, results flips the correspond exactly to certain intersections of the Voronoi cells before and after the removal.
It followed that the flips per point was bounded by the aspect of the Voronoi cells.
In our case the aspect is bounded by a constant and thus we get a constant number of flips.

\paragraph*{The Cost of Voronoi Refinement}
The principle that drives the efficiency of the Sparse Voronoi Refinement algorithm of Hudson et al.~\cite{hudson06sparse} is to maintain the well-spaced condition as a strict algorithmic invariant.
This is in contrast to classic Delaunay refinement in which the input points are all added first.
Our algorithm resembles a mix between these approaches, achieving the efficiency of the former with the simplicity of the latter.
Once per wave, we add Steiner points.
Because the point set was well-spaced at the start of the wave and the points added in a wave form a net, the points are still well-spaced at the end of the wave (albeit with a larger constant).
So, we are effectively running a standard Voronoi refinement on an input that is already nearly refined.
Having a bound on the spacing of the points implies that every Steiner point can be added in constant time (again, this was a major lesson from Hudson et al.~\cite{hudson06sparse}).
Thus, the total cost of all the refinement steps in the algorithm will be linear in the total number of Steiner points.

According to the standard size analysis of refinement algorithms, the total number of points will be asymptotically optimal~\cite{ruppert95delaunay} and will be at most $O(n\log \spread)$ (see Hudson et al.~\cite{hudson06sparse}).
However, for many point sets, this bound is loose and it has been shown that for a broad class of inputs, the number of Steiner points will be $O(n)$~\cite{miller08linear,hudson09size,sheehy12new}.
So, the worst case running time of the refinement step is $O(n\log \spread)$, but one should not be surprised to see output sizes closer to $O(n)$ for real examples.

\end{document}